\documentclass[12]{amsart}

\usepackage{amssymb}
\makeatletter
\numberwithin{equation}{section}
\numberwithin{figure}{section}
\numberwithin{table}{section}
\newtheorem{thm}{Theorem}[section]
\newtheorem{lemma}[thm]{Lemma}
\newtheorem{cor}[thm]{Corollary}
\newtheorem{pro}[thm]{Proposition}
\newtheorem{defn}[thm]{Definition}

\newtheorem{rem}[thm]{Remark}

\newcommand{\cal}{\mathcal}

\begin{document}
\title[Set-theoretic YBE,  braces and Drinfeld twists]{Set-theoretic Yang-Baxter equation,  braces  and  Drinfeld twists }
\author[Anastasia Doikou]{Anastasia Doikou}

\address[A. Doikou] {Dept of Mathematics, Heriot-Watt University,
Edinburgh EH14 4AS, and Maxwell Institute for Mathematical Sciences, Edinburgh}
\email{A.Doikou@hw.ac.uk}


\date{\today}

\begin{abstract}
We consider involutive, non-degenerate, finite set-theoretic solutions of the Yang-Baxter equation.  
Such solutions can be always obtained using certain algebraic structures 
that generalize nilpotent rings called braces. Our main aim here is to express such solutions in terms of admissible 
{\it Drinfeld twists} substantially extending recent preliminary results. 
We first identify the generic form of the twists associated to set-theoretic solutions and  we show 
that these twists are admissible, i.e. they satisfy  a certain co-cycle condition. These findings 
are also valid for Baxterized solutions of the Yang-Baxter equation constructed from the set-theoretical ones.

\end{abstract}

\maketitle

\date{}
\vskip 0.4in

\section{Introduction}

\noindent 
The idea of set-theoretic solutions of the Yang-Baxter equation (YBE \cite{Baxter}) was first suggested by Drinfield 
\cite{Drin} and since then there has been a considerable research activity on this topic from the algebraic point of view 
\cite{Hienta, ESS, Eti}, but also in the context of classical integrable systems \cite{ABS, Veselov, Papag}. 
More recently, an algebraic structure that generalizes nilpotent  rings, called a brace, was introduced \cite{[25], [26]}, to describe 
all finite, involutive, non-degenerate set-theoretic solutions of  the Yang-Baxter equation.   
Rump showed \cite{[25], [26]} that every brace yields a solution to the Yang-Baxter equation, and every 
non-degenerate, involutive set-theoretic  solution of the Yang-Baxter equation can be obtained from a brace.
Subsequently skew-braces were developed in \cite{GV} to describe non-involutive solutions. 
This emerging research area has been particularly fruitful and numerous relevant studies have 
been produced  over the past few years (see for instance
 \cite{bcjo, Catino, [6]}, \cite{Gateva}-\cite{15}, \cite{GV, JKA, SVB, LAA}).

In \cite{DoiSmo1, DoiSmo2} key links between set-theoretical solutions coming from braces and quantum
integrable systems and the associated quantum algebras were uncovered. 
More precisely:
\begin{enumerate}
\item  Quantum groups associated to Baxterized solutions of the Yang-Baxter equations coming 
from braces were derived via the FRT construction \cite{FadTakRes}.
\item  Novel classes of quantum discrete integrable systems with periodic and open boundary conditions were produced.
 \item Symmetries of the periodic and open transfer matrices of the novel integrable systems were identified.
\item Preliminary findings on Drinfeld twists for set-theoretic solutions were presented.
\end{enumerate}
Note that in \cite{ESS} Etingof, Schedler and Soloviev constructed quantum groups
associated to set-theoretic solutions, however in \cite{DoiSmo1} a different construction is used coming
from parameter dependent solutions of the Yang-Baxter equation and thus the corresponding 
quantum groups differ from those in \cite{ESS}.

Our main focus here is the study of involutive  set-theoretic solutions coming from braces, 
based on the findings presented in \cite{DoiSmo1, DoiSmo2}.
More precisely, our aim is to identify the explicit form of so called admissible 
Drinfeld twists for brace solutions of the YBE.
We note that  in Drinfeld's original works on quantum algebras 
\cite{Drinfeld, Drintw} it was required that the universal ${\cal R}$-matrix, 
solution of the Yang-Baxter equation has a semi-classical limit, 
i.e. it can be expressed as formal series expansion,
${\cal R} = {\bf 1} +\sum_{n =1}^{\infty} h^{n} {\cal R}^{(n)}$, where $h$ is some ``deformation'' parameter. 
In \cite{Drintw} Drinfeld introduced the notion of the twist ${\mathfrak F}$, which links distinct solutions 
of the Yang-Baxter equation, and consequently yields distinct quasi triangular Hopf or (quasi) Hopf algebras,
provided that the twist is {\it admissible}. Such twists have also  semi-classical limits \cite{Drintw}, i.e. 
they can be expressed as  formal series expansions in powers of the deformation parameter, 
with the leading term being the identity.

In the analysis of \cite{DoiSmo1, DoiSmo2} {\it Baxterized} $R$-matrices coming from set-theoretic solutions 
of the Yang-Baxter equation were identified,  being of the form
$R(\lambda) = r+ {1\over \lambda} {\cal P}$, where $r$ is the set-theoretic solution of the
Yang-Baxter equation and ${\cal P}$ is the permutation operator. 
Interestingly the $r$-matrix does not contain any free parameter (deformation parameter), 
and consequently the $R$-matrix has no semi-classical analogue. 
This is the pivotal difference between our study here and  Drinfeld's analysis  \cite{Drintw}.
A similar observation can be made about the associated Drinfeld twist, which is explicitly identified in the present study.
To be more precise the main results of this investigation, presented in section 3, are: 
\begin{enumerate}
\item The identification of the explicit form 
of ``local'' (2-site) and ``global'' ($n$-site)  Drinfeld twists.
\item The proof of the admissibility of Drinfeld's twists associated to involutive, non-degenerate, 
set-theoretic solutions of the Yang-Baxter equation. 
\end{enumerate}

\subsection*{The outline of the paper} We describe below what is achieved in each of the subsequent sections.

\begin{itemize}

\item In section 2 we present necessary preliminary notions regarding the set-theoretic solutions of the YBE 
and the associated quantum groups derived via the FRT construction \cite{DoiSmo1}.  More specifically:

\begin{enumerate}

\item In subsection 2.1 we recall the definition of the set-theoretic YBE.  We also recall Rump's theorem, which states that
every involutive, non-degenerate, set-theoretic solution of the YBE comes from a brace and vice versa.

\item In subsection 2.2 we review the construction of the quantum group associated 
to involutive, set-theoretic solutions by means of the FRT formulation.
\end{enumerate}

\item  In section 3 we first  recall  basic definitions and results on Hopf algebras and Drinfeld twists \cite{Drinfeld, Drintw}.
We then move on with the presentation of our main findings regarding the derivation of admissible Drinfeld twists
associated to set-theoretic solutions of the YBE.  Specifically:

\begin{enumerate}
\item In subsection 3.1 we review basic background on bialgebras, Hopf and quasi-triangular Hopf algebras.  
We then introduce the notion of admissible Drinfeld twist and recall 
fundamental propositions on admissible twists adjusted for your purposes here. 

\item In subsection  3.2 we present the new findings on the Drinfeld twists for involutive, non-degenerate, set-theoretic solutions. 
To achieve this we rely on the results of \cite{DoiSmo2}, where the set-theoretic braid $r$-matrix 
was obtained  from the permutation operator via a similarity transformation. We use the similarity transformation to derive explicit forms
for the set-theoretic twist. We also derive co-product expressions of the twist, which enable the derivation of the $n$-twist.
From the derived explicit expressions we are able to show the admissibility of the twist i.e. 
we show that it satisfies the co-cycle condition.
Note that we introduce the term {\it quasi-admissible twist} given
that we only show the admissibility of the set-theoretic twists, 
without exploring the notions of the co-unit and antipode, 
although some related comments are presented at the end of subsection 3.3.

\item In subsection 3.3 we focus on a simple, but characteristic example of set-theoretic solution 
of the YBE known as Lyubashenko's solution \cite{Drin}). We first introduce this class of solutions 
and we show that they
can be expressed as {\it Reshetikhin type twists}  recalling the results of \cite{DoiSmo2}. 
We then move on to show that this is an admissible twist and we derive simple expressions for the $n$-twist.
At the end of this subsection  we present some preliminary observations related to
the  action of the co-unit on the twisted co-products. We employ the Lyubashenko 
solution to illustrate the non-trivial action of the co-unit.
\end{enumerate}

\end{itemize}

\section{Preliminaries}

\noindent We present in this section basic background information regarding set-theoretic solutions of the 
Yang-Baxter equation and braces as well as a brief review on the recent findings of \cite{DoiSmo1}
on the links between set-theoretic solutions of the Yang-Baxter equation from braces and quantum algebras.

\subsection{The set-theoretic Yang-Baxter equation}
\noindent Let $X=\{{\mathrm x}_{1},{\mathrm x}_{2}, \ldots, {\mathrm x}_{\cal N}\}$ be a set and ${\check r}:X\times X\rightarrow X\times X,$ such that
 \[{\check r}(x,y)= \big (\sigma _{x}(y), \tau _{y}(x)\big ).\] 
We say that $\check r$ is non-degenerate if $\sigma _{x}$ 
and $\tau _{y}$ are bijective functions. Also, the solutions $(X, \check r)$  is involutive: $\check r ( \sigma _{x}(y), \tau _{y}(x)) = 
(x, y)$, ($\check r \check r (x, y) = (x, y)$). We focus on non-degenerate, involutive solutions of the set-theoretic braid equation:
\[({\check r}\times id_{X})(id_{X}\times {\check r})({\check r}\times id_{X})=(id_{X}\times {\check r})({\check r}
\times id_{X})(id_{X}\times {\check r}).\]

Let $V$ be the space of dimension equal to the cardinality of $X$, and with a slight abuse of notation, 
let  $\check r$ also  denote the $R$-matrix associated to the linearisation of ${\check r}$ on 
$V={\mathbb C }X$ (see \cite{LAA} for more details), i.e.
$\check r$  is the ${\mathcal N}^2 \times{\mathcal N}^2$ matrix: 
\begin{equation}
\check r = \sum_{x,y,z,w \in X} \check r(x,z|y,w) e_{x,z} \otimes e_{y, w}, \label{set0}
\end{equation}
where
$e_{x, y}$ is the ${\mathcal N} \times {\mathcal N}$  matrix: $(e_{x,y})_{z,w} =\delta_{x,z}\delta_{y,w} $.
Then for the $\check r$-matrix related  to $(X,{\check r})$:  ${\check r}(x,z|y,w)=\delta_{z, \sigma_x(y)} 
\delta_{w, \tau_y(x)}$.
 Notice that the matrix $\check r:V\otimes V\rightarrow V\otimes V$ satisfies the (constant) Braid equation:
\[({\check r}\otimes I_{V})(I_{V}\otimes {\check r})({\check r}\otimes I_{V})=(I_{V}\otimes {\check r})({\check r}\otimes 
I_{V})(I_{V}\otimes {\check r}).\]
Notice also that ${\check r}^{2}=I_{V \otimes V}$ the identity matrix, because $\check r$ is involutive.

For set-theoretic solutions it is thus convenient to use the matrix notation:
\begin{equation}
{\check r}=\sum_{x, y\in X} e_{x, \sigma_x(y)}\otimes e_{y, \tau_y(x)}. \label{brace1}
\end{equation}
Define also, $r={\mathcal P}{\check r}$, where ${\mathcal P} = \sum_{x, y \in X} e_{x,y} \otimes e_{y,x}$ 
is the permutation operator,  consequently
${ r}=\sum_{x, y\in X} e_{y,\sigma_x(y)}\otimes e_{x, \tau_y(x)}.$
The  Yangian \cite{Yang} is a special case: $\check r(x,z|y,w)= \delta_{z,y}\delta_{w,x} $.

$ $

Let us now recall the role of braces in the derivation of set-theoretic solutions of the Yang-Baxter equation.
In \cite{[25], [26]} Rump showed that every solution $(X, \check r)$ can be in a good way embedded in a brace.
\begin{defn}[Proposition $4$, \cite{[26]}]
A {\em left brace } is an abelian group $(A; +)$ together with a 
multiplication $\cdot $ such that the circle operation $a \circ  b =
 a\cdot b+a+b$ makes $A$ into a group, and $a\cdot (b+c)=a\cdot b+a\cdot c$.
\end{defn} 
 In many papers, an equivalent definition is used \cite{[6]}.

 The additive identity of a brace $A$ will be denoted by $0$ and the multiplicative identity by $1$. In every brace $0=1$. 
The same notation will be used for skew braces (in every skew brace $0=1$).

Throughout this paper we will use the following result,  which  is  implicit in \cite{[25], [26]} 
and explicit in Theorem 4.4 of \cite{[6]}.

\begin{thm}\label{Rump}(Rump's theorem, \cite{[25], [26], [6]}).  It is known that for an involutive, non-degenerate solution  
of the braid equation there is always an underlying brace $(B, \circ , +)$, 
such that the maps $\sigma _{x}$ and $\tau _{y}$ come from this brace, and $X$ is a subset in this brace such that ${\check r}(X,X)
\subseteq (X,X)$ and ${\check r}(x,y)=(\sigma _{x}(y), \tau _{y}(x))$, where $\sigma _{x}(y)=x\circ y-x$, $\tau _{y}(x)=t\circ x-t$, 
where $t $ is the inverse of $\sigma _{x}(y)$ in the circle group $(B, \circ )$.  Moreover, we can assume that every element from $B$ 
belongs to the additive group $(X, +)$  generated by elements of $X$. In addition every solution of this type is a non-degenerate, involutive 
set-theoretic solution of the braid equation.
\end{thm}

We will call the brace $B$  an underlying brace of the solution $(X,{\check r})$, or a brace associated to the  
solution $(X,{\check r})$. We will also say that the solution $(X, \check r )$ is associated to brace $B$. 
Notice that this is also related  to the formula of
 set-theoretic solutions associated to the braided group (see \cite{ESS} and \cite{gateva}).

\subsection{Yang-Baxter equation $\&$ quantum groups}

\noindent In this subsection we briefly review the main results reported in \cite{DoiSmo1} on the various links between braces, 
representations of the $A$-type Hecke algebras, 
and quantum algebras.

Recall first the Yang-Baxter equation in the braid form in the presence of spectral parameters $\lambda_1,\ \lambda_2$ ($\delta = \lambda_1 - \lambda_2$):
\begin{equation}
\check R_{12}(\delta)\ \check R_{23}(\lambda_1)\ \check R_{12}(\lambda_2) = \check R_{23}(\lambda_2)\
 \check R_{12}(\lambda_1)\ \check R_{23}(\delta) . \label{YBE1}
\end{equation}
where $\check R: V \to V,$  and let in general $\check R = \sum_{j} a_j \otimes b_j,$ then in the index notation $\check R_{12} =\sum_j a_j \otimes b_j \otimes I_V,$ 
 $\check R_{23} =\sum_j  I_V \otimes a_j \otimes b_j $ and  $\check R_{13} =\sum_j a_j \otimes  I_V \otimes b_j.$

We focus here on brace  solutions of the YBE, given by (\ref{brace1}). 
The brace solution\footnote{All, finite, non-degenerate, involutive, 
set-theoretic solutions of the YBE are coming from braces,
therefore we will call such solutions {\it  brace solutions}.} $\check r$ (\ref{brace1})
is a representation of the $A$-type Hecke algebra for $q=1$ (see also \cite{DoiSmo1}),
as $\check r$ satisfies the braid relations and also 
$\check r^2 = I_{X \otimes X}.$ Due to the fact that the set-theoretic $\check r$ provides a representation 
of the $A$-type Hecke algebra 
Baxterized solutions of the Yang-Baxter equation can be derived \cite{DoiSmo1}:
\begin{equation}
\check R(\lambda) = \lambda \check r + I_X\otimes I_X, \label{braid1}
\end{equation}
where $I_{X}$ is the identity matrix of dimension equal to the cardinality of the set $X$.

Let also $R = {\cal P} \check R$, (recall the permutation operator ${\cal P} = \sum_{x, y} e_{x,y}\otimes e_{y,x}$),  
then the following basic properties for $R$ matrices coming from braces were shown in \cite{DoiSmo1}:\\

\noindent {\bf Basic Properties.} {\em  
The brace $R$-matrix satisfies the following fundamental properties:}
\begin{eqnarray}
&&  R_{12}(\lambda)\  R_{21}(-\lambda) = (-\lambda^2 +1) I_{X\otimes X}, 
~~~~~~~~~~~~~
\mbox{{\it Unitarity}} \label{u1}\\
&&  R_{12}^{t_1}(\lambda)\ R_{12}^{t_2}(-\lambda -{\cal N}) = 
\lambda(-\lambda -{\cal N})I_{X\otimes X}, 
~~~~~\mbox{{\it Crossing-unitarity}} \label{u2}\\
&& R_{12}^{t_1 t_2}(\lambda) = R_{21}(\lambda), \label{tt}\nonumber
\end{eqnarray}
{\it where $^{t_{1,2}}$ denotes transposition on the first, second space respectively.}

$ $

\noindent{\bf The quantum algebra.}  Our approach on deriving the quantum groups associated to set-theoretic solutions \cite{DoiSmo1, DoiSmo2}
is based on the FRT construction \cite{FadTakRes}, which is somehow dual to the Hopf algebraic description \cite{Drinfeld}.
Indeed, the FRT construction can be considered as an inverse procedure of the one that uses 
the quasi-triangular Hopf algebra axioms in obtaining a solution of the YBE. 
The FRT construction considers a solution $R:  V \otimes V\to V \otimes  V$  ($V$ 
is usually a finite vector space) of the YBE as an input and produces a bialgebra as an output.

Given a solution of the Yang-Baxter equation, the quantum algebra is defined via the fundamental relation \cite{FadTakRes}
(we have multiplied the familiar RTT relation with the permutation operator):
\begin{equation}
\check R_{12}(\lambda_1 -\lambda_2)\ L_1(\lambda_1)\ L_2(\lambda_2) = L_1(\lambda_2)\ L_2(\lambda_1)\ 
\check R_{12}(\lambda_1 -\lambda_2). \label{RTT}
\end{equation}
$\check R(\lambda) \in \mbox{End}({\mathbb C}^{\cal N}) \otimes  \mbox{End}({\mathbb C}^{\cal N})$, $\ L(\lambda) \in 
\mbox{End}({\mathbb C}^{\cal N} ) \otimes {\mathfrak A}$, where ${\mathfrak A}$ is the quantum algebra defined by (\ref{RTT}). 
We focus on solutions associated to braces given by (\ref{braid1}), (\ref{brace1}). 
The defining relations of the corresponding quantum algebra were derived in \cite{DoiSmo1}.

\noindent {\it The quantum algebra associated to the brace $R$ matrix  (\ref{braid1}), (\ref{brace1}) 
is defined by generators $L^{(m)}_{zw},\ z, w \in X$, and defining relations }
\begin{eqnarray}
L_{z,w}^{(n)} L_{\hat z, \hat w}^{(m)} - L_{z,w}^{(m)} L_{\hat z, \hat w}^{(n)} &=& 
L^{(m)}_{z, \sigma_w(\hat w)} L^{(n+1)}_{\hat z,\tau_{\hat w}( w)}- L^{(m+1)}_{z, \sigma_w(\hat w)} 
L^{(n)}_{\hat z, \tau_{\hat w}( w)}\nonumber\\ &-& L^{(n+1)}_{ \sigma_z(\hat z),w} 
L^{(m)}_{\tau_{\hat z}( z), \hat w }+ L^{(n)}_{ \sigma_z(\hat z, )w}  L^{(m+1)}_{\tau_{\hat z}( z), \hat w}. \label{fund2}
\end{eqnarray}

The proof is based on the fundamental relation (\ref{RTT}) and the form of the brace $R$- matrix
(for the detailed proof see \cite{DoiSmo1}). Recall also that
in the index notation we define $\check R_{12} = \check R \otimes \mbox{id}_{\mathfrak A}$:
\begin{eqnarray}
&&  L_1(\lambda) = \sum_{z, w \in X} e_{z,w} \otimes I \otimes L_{z,w}(\lambda),\ \quad  L_2(\lambda)= \sum_{z, w \in X}
I  \otimes  e_{z,w}  \otimes L_{z,w}(\lambda)  \label{def}
\end{eqnarray} where $L_{z,w}(\lambda) = \sum_{n=0}^{\infty}\lambda^{-n}L_{z,w}^{(n)}$ and $L_{z,w}^{(n)}$
are the generators of the affine algebra ${\mathfrak A}$ and $\check R$ is given in (\ref{braid1}), (\ref{brace1}). 
The quantum algebra is equipped with a co-product $\Delta: {\mathfrak A} 
\to {\mathfrak A} \otimes {\mathfrak A}$ \cite{FadTakRes, Drinfeld}. Indeed, we define 
$ {\mathrm T}_{1,23}(\lambda)= L_{13}(\lambda) L_{12}(\lambda),$
which satisfies (\ref{RTT}) and is expressed as ${\mathrm T}_{1,23}(\lambda) = \sum_{x,y \in X} e_{x,y} \otimes \Delta(L_{x,y}(\lambda)).$

\section{Set-theoretic solutions as Drinfeld twists}

\noindent  In this section we recall  basic definitions and results on Hopf algebras and Drinfeld twists \cite{Drinfeld, Drintw}, 
we then move on with the presentation of our main findings regarding the derivation of admissible Drinfeld twists
associated to set-theoretic solutions of the YBE.

\subsection{Quasi-triangular Hopf algebras $\&$ Drinfeld twists}

\noindent Before we present our brief review on Drinfeld's twists let us first recall the notion of a quasi-triangular Hopf algebra.
The known quantum algebras defined by the RTT relation (\ref{RTT}) are 
{\it quasi-triangular Hopf algebras} (on a detailed discussion on bialgebras, Hopf algebras, and quasi-triangular 
(quasi) Hopf algebras  we refer the interested reader for instance in \cite{Drinfeld, Drintw, ChaPre, Majid,  MaiSan}). 
We give below some useful definitions  regarding, bialgebras, Hopf and
quasi-triangular Hopf algebras (see also for instance \cite{ChaPre, Majid}).

\begin{defn}{\label{def0}} A bialgebra $H$ is a vector space over some field $k$ with linear maps:
\begin{enumerate}
\item $m: H\otimes H \to H $,  multiplication, $m(a\otimes b) = ab$, which is associative, i.e. $(ab)c = a(bc)$ $\forall a,b,c \in H.$
\item $\eta:  k \to H$ such that it produces the unit element for $H$, $\eta(1) = {\bf 1}_{H}$ ($\eta(c) = c \cdot  {\bf 1}_{H}$).
\item $\Delta: H \to H \otimes H$, coproduct $\Delta(a) = \sum_j a_j \otimes b_j,$ which is coassociative, i.e. 
$(\Delta \otimes \mbox{id})\Delta = (\mbox{id} \otimes \Delta)\Delta.$
\item $\epsilon: H \to k$, counit such that $(\epsilon \otimes \mbox{id})\Delta(a) = (\mbox{id} \otimes \epsilon)\Delta(a) = a$, $\forall a\in H$.
\item $\Delta,\ \epsilon$ are algebra homomorphisms  and $H \otimes H$ has the structure of a tensor product algebra: 
$(a \otimes b)(c\otimes d)=ac \otimes bd$, $\forall a,b,c,d \in H.$
\end{enumerate}
\end{defn}
In other words, $(H,\ m,\ \eta)$ is an associative algebra and $(H,\Delta, \epsilon)$ is an associative coalgebra.
The compatibility of the bialgebra  axioms is typically represented by commutative 
diagrams (see for instance \cite{ChaPre, Majid} and references therein). 

\begin{defn}
A Hopf algebra $H$ is a bialgebra equipped  with a bijective linear map, the ``antipode'', $s: H \to H,$ such that 
$m\big ((s\otimes \mbox{id})\Delta(g) \big)=m\big ((\mbox{id}\otimes s)\Delta(g)\big ) = \epsilon(g)\cdot  {\bf 1}_H,$ $\forall g \in H.$
\end{defn}

We may  now give the definition of a quasi-triangular Hopf algebra, which is our main focus in this subsection.

\begin{defn}{\label{def02}} Let ${\mathfrak A}$ be a  Hopf  algebra over some field $k$, then
${\mathfrak A}$ is a quasi-triangular Hopf algebra if there exists an invertible element ${\cal R}\in {\cal A} \otimes {\cal A}$:

\begin{enumerate}
\item  ${\cal R}\ \Delta(a) = \Delta^{op}(a)\ {\cal R}$, $\ \forall a \in {\mathfrak A}$,\\
where $\Delta: {\mathfrak A} \to {\mathfrak A} \otimes {\mathfrak A}$ is the co-product on ${\mathfrak A}$ and 
$\Delta^{op}(a) = \pi \circ \Delta(a)$, $\pi: 
{\mathfrak A} \to {\mathfrak A}$ such that $\pi(a\otimes b) = b \otimes a$. 

\item $(\mbox{id} \otimes \Delta){\cal R} = {\cal R}_{13} {\cal R}_{12},$ 
and $ (\Delta \otimes \mbox{id}) {\cal R} = {\cal R}_{13}{\cal R}_{23}$.
\end{enumerate}
\end{defn}

Also, the following statements hold:
\begin{itemize}
 \item The antipode $s: {\mathfrak A} \to {\mathfrak A}$ satisfies  $(\mbox{id} \otimes s) {\cal R}^{-1} 
= {\cal R}$, $(s \otimes \mbox{id} ){\cal R} = {\cal R}^{-1}$, 
\item The  co-unit  $\epsilon: {\mathfrak A} \to k$  satisfies
$(\mbox{id} \otimes \epsilon ) {\cal R} = ( \epsilon \otimes \mbox{id} ) {\cal R} = {\bf 1} $.

\item Due to  (1) and (2) of Definition \ref{def02} the ${\cal R}$-matrix satisfies the Yang-Baxter equation 
${\cal R}_{12} {\cal R}_{13} {\cal R}_{23} ={\cal R}_{23} {\cal R}_{13} {\cal R}_{12}$. 
\end{itemize}
Proofs of the above statements can be found for instance in \cite{ChaPre, Majid}. 
Note that a sufficient condition for the co-associativity of $\Delta^{(op)}$ is 
$({\cal R} \otimes  {\bf 1})\ (\Delta \otimes \mbox{id}){\cal R} = ( {\bf 1} \otimes {\cal R})\ (\mbox{id}\otimes\Delta){\cal R}$, 
i.e. ${\cal R}$ should satisfy the co-cycle condition. The co-cycle condition together with  (2) of Definition \ref{def02} lead to YBE.

The index notation holds in this case as follows.  Let ${\cal R} = \sum_j a_j\otimes b_j\in {\mathfrak A} \otimes {\mathfrak A}$ 
then ${\cal R}_{12}=  \sum_j a_j\otimes b_j \otimes  {\bf 1}$,  $\ {\cal R}_{23}=  \sum_j   {\bf 1}\otimes a_j\otimes b_j $ etc.  

If in addition to the above conditions ${\cal R} {\cal R}^{(op)}= {\bf 1}_{{\cal A} \otimes {\cal A}}$, where ${\cal R}^{(op)} = 
\sigma\circ {\cal R}$, (or in the index notation ${\cal R}_{12} {\cal R}_{21} = {\bf 1}$), then ${\mathfrak A}$ is a {\it  triangular Hopf-algebra}

\subsection*{Drinfelds twists}
In \cite{Drinfeld, Drintw} Drinfeld introduced
the notion of twisting (or deformation) for (quasi)-Hopf algebras and quasi-triangular (quasi)-Hopf
algebras. Recall that we focus in this subsection on quasi-triangular Hopf algebras.
 Let ${\mathfrak A}$, equipped with $(m, \eta, \epsilon, \Delta,\ s,\ \epsilon)$, be a quasi-triangular Hopf algebra over ${\mathbb C}$ 
and ${\mathfrak F}\in {\mathfrak A} \otimes {\mathfrak A}$ be an invertible
element such that $(\mbox{id} \otimes \epsilon){\mathfrak F} = (\epsilon \otimes \mbox{id}){\mathfrak F} = {\bf 1}$.
The twist of ${\mathfrak A}$ generated by ${\mathfrak F}$ defines a new
quasi-triangular Hopf algebra $\tilde {\mathfrak A}$,  equipped with $(m, \eta, \epsilon, \tilde \Delta,\ \tilde s,\ \tilde \epsilon),$ 
having the same elements and product law
as ${\mathfrak A}$, and such that:
\begin{equation}
\tilde {\cal R} = {\mathfrak F}^{op} {\cal R}\  {\mathfrak F}^{-1},
\end{equation}
where we define ${\mathfrak F}^{op} = \pi \circ {\mathfrak F}$, and the permutation (flip) map $\pi: 
a\otimes b \mapsto b \otimes a$.
In the index notation ${\mathfrak F}_{12}^{op} = {\mathfrak F}_{21}$. 
Also, $u=m\big ( (\mbox{id} \otimes s){\mathfrak F}\big )$ is an invertible element in ${\mathfrak A}$ with 
$ u^{-1} =m\big (( s \otimes \mbox{id} ){\mathfrak F}^{-1}\big )$,
and the twisted maps, $\tilde \Delta$ and $\tilde s$ are defined as
\begin{equation}
\tilde \Delta(a) = {\mathfrak F} \Delta(a) {\mathfrak F}^{-1}, \quad \tilde s = u s(a) u^{-1} .
\end{equation}

\begin{defn} {\label{def2}}  A twist ${\mathfrak F}$ is called admissible if it  satisfies the co-cycle condition:
\begin{equation}
( {\mathfrak F}\otimes  {\bf 1}) \ (\Delta \otimes \mbox{id})  {\mathfrak F} =  
( {\bf 1} \otimes  {\mathfrak F})\ (\mbox{id} \otimes \Delta)  {\mathfrak F}.
\end{equation}
\end{defn}
In the special case that $\tilde {\cal R} = {\bf 1}_{{\cal A} \otimes {\cal A}},$ then ${\cal R} = ({\mathfrak F}^{(op)})^{-1}{\mathfrak F}$ 
and such a twist is called {\it factorizable}.

We now introduce some convenient index notation:
\begin{equation}
{\mathfrak F}_{12,3}:= (\Delta \otimes \mbox{id}) {\mathfrak F} , \qquad {\mathfrak F}_{1,23}:=  
(\mbox{id} \otimes \Delta ) {\mathfrak F}.
\end{equation}
Then according to the co-cycle condition we introduce,  
$ {\mathfrak F}_{123} =  {\mathfrak F}_{12}  {\mathfrak F}_{12,3} =  {\mathfrak F}_{23} {\mathfrak F}_{1.23}$. 
In general, we define
\begin{equation}
{\mathfrak F}_{12...n, n+1}:= (\Delta^{(n)} \otimes \mbox{id}) {\mathfrak F} , 
\qquad {\mathfrak F}_{1,2... n+1}:=  (\mbox{id} \otimes \Delta^{(n)}) 
{\mathfrak F}  \label{codef}
\end{equation}
and consequently we define the $n$-twist, compatible also with the co-cycle condition:
\begin{equation}
{\mathfrak F}_{12...n} = {\mathfrak F}_{12...n-1} {\mathfrak F}_{12..n-1,n} =
{\mathfrak F}_{23..n} {\mathfrak F}_{1,2..n}.
\end{equation}
We introduce a similar notation for tensor representations of the quantum algebra 
\begin{equation}
 {\mathfrak T}_{0,12..n}:= (\mbox{id} \otimes \Delta^{(n)}) {\cal R}  
\end{equation}
where recall  $(\mbox{id} \otimes \Delta){\cal R} = {\cal R}_{02} {\cal R}_{01}$ 
and $(\Delta \otimes \mbox{id}) = {\cal R}_{10} {\cal R}_{20},$ thus
we can explicitly write ${\mathfrak T}_{0,12...n} = {\cal R}_{0n} \ldots {\cal R}_{02} {\cal R}_{01}$. The $n+1$ co-cycle condition 
and the explicit expressions of the $n+1$ twist are also useful for our purposes here:
\begin{eqnarray}
 {\mathfrak F}_{012...n}  &=&  {\mathfrak F}_{01}  {\mathfrak F}_{01,2}  {\mathfrak F}_{012,3} \ldots  {\mathfrak F}_{01...n-1,n}\nonumber \\
& = & {\mathfrak F}_{n-1n}  {\mathfrak F}_{n-2, n-1n} \ldots  {\mathfrak F}_{0,1..n-1 n}. \label{gencoc}
\end{eqnarray}

We state below three Propositions that include some of the main Drinfeld's results \cite{Drinfeld, Drintw} restricted to ${\cal R}$-matrices 
that satisfy the Yang-Baxter equation. These results are most relevant to our present analysis 
that follows in the next subsection. We also refer the interested reader to \cite{MaiSan} 
and references therein on similar proofs using the index notation.

\begin{pro}{(Drinfled 1)\label{lemma1}} Let ${\cal R}$ be a solution of the Yang-Baxter equation 
and ${\mathfrak F}$ be an admissible twist such that, in the index notation, 
$\tilde {\cal R}_{0n}= {\mathfrak F}_{n0}\  
{\cal R}_{0n}\ {\mathfrak F}^{-1}_{0n},$\\  
where $n \in \{1, \ldots N \}$, then
\begin{equation}
 {\mathfrak F}_{1... n-1n0 n+1...N} {\cal R}_{0n} =
\tilde {\cal R}_{0n}{\mathfrak F}_{1... n-10n n+1...N}.
\end{equation}
\end{pro}
\begin{proof}
The proof relies on three basic statements:
\begin{enumerate}
\item $\tilde {\cal R}_{0n}= {\mathfrak F}_{n0}\  {\cal R}_{0n}\ {\mathfrak F}^{-1}_{0n} $
\item ${\mathfrak F}_{m-1, m...n 0 n+1...N} {\cal R}_{0n} =
{\cal R}_{0n}{\mathfrak F}_{m-1, m...n-1 0 n...N}$,\\ $\ m\leq n \in \{1, \ldots N\}$
and\\ $ {\mathfrak F}_{12...n0n+1...m, m+1} {\cal R}_{0n} =
{\cal R}_{0n}{\mathfrak F}_{12...n-10n...m, m+1}, $ \\ $\ m\geq n \in \{1, \ldots N\}$
\item the generalized co-cycle condition (\ref{gencoc})
\end{enumerate}
Statement (2) is a natural consequence  of (\ref{codef}) and the fact that ${\cal R}$ 
satisfies the Yang-Baxter equation or equivalently conditions (1) and (2) of Definition \ref{def02}.

We start with
\begin{eqnarray}
&& {\mathfrak F}_{1... n-1n0 n+1...N} {\cal R}_{0n}  = \nonumber\\
&& {\mathfrak F}_{1... n-1n0} {\mathfrak F}_{1... n-1n0, n+1} \ldots 
{\mathfrak F}_{1... n-1n0 n+1...N-1,N} {\cal R}_{0n}=\nonumber\\
&&  {\mathfrak F}_{1... n-1n0} {\cal R}_{0n} {\mathfrak F}_{1... n-10n, n+1} \ldots 
{\mathfrak F}_{1... n-1 0 n n+1...N-1,N}, \label{exp1}
\end{eqnarray}
where we have repeatedly used the generalized expressions (\ref{gencoc}) and statement (2).
We now focus on 
\begin{eqnarray}
{\mathfrak F}_{1... n-1n0} {\cal R}_{0n}  &=& 
{\mathfrak F}_{n0} {\mathfrak F}_{n-1, n0}\ldots 
{\mathfrak F}_{1, 2...n0}  {\cal R}_{0n} \nonumber\\ 
&=&  {\mathfrak F}_{n0} {\cal R}_{0n}  {\mathfrak F}_{n-1, 0n}\ldots 
{\mathfrak F}_{1, 2...0n}  \nonumber\\
& = &  \tilde {\cal R}_{0n}   {\mathfrak F}_{0n}
 {\mathfrak F}_{n-1, 0n}\ldots {\mathfrak F}_{1, 2...0n}  \nonumber\\
&=& \tilde  {\cal R}_{0n}   {\mathfrak F}_{12...0n}. \label{exp20}
\end{eqnarray}
Then by means of (\ref{exp20}) statement (1) and (\ref{gencoc}) expression (\ref{exp1}) becomes\\
$\tilde  {\cal R}_{0n} {\mathfrak F}_{1... n-1 0 n n+1...N},$
which concludes our proof.
\end{proof}

We may now present  the next two basic Propositions.
\begin{pro}{(Drinfeld 2) \label{Drin1}} Let ${\cal R}$ satisfy the Yang-Baxter equation and 
$\tilde{\cal R} ={\mathfrak  F}^{op}\  {\cal R}\ {\mathfrak  F}^{-1}$,
 where ${\mathfrak F}$ is an admissible twist, then $\tilde {\cal R}$ also satisfies the YBE.
\end{pro}
\begin{proof}
We are employing the index notation together with Proposition \ref{lemma1}. ${\cal R}$ satisfies the YBE, 
we multiply YBE with ${\mathfrak F}_{321}$ form the left and ${\mathfrak F}_{123}^{-1}$ from the right:
\begin{equation}
{\mathfrak F}_{321}\ ( {\cal R}_{12}\ {\cal R}_{13}\  {\cal R}_{23}=  {\cal R}_{23}\
 {\cal R}_{13}\ {\cal R}_{12} )\ {\mathfrak F}_{123}^{-1}. \label{trans}
\end{equation}
We focus on the left side of the equation above, and use Proposition \ref{lemma1}:
\begin{eqnarray}
&&    {\mathfrak F}_{321}\  {\cal R}_{12}\  {\cal R}_{13}\  {\cal R}_{23}\  {\mathfrak F}_{123}^{-1} =  
\tilde  {\cal R}_{12}\  {\mathfrak F}_{312 }\ {\cal R}_{13}\  {\cal R}_{23}\  {\mathfrak F}_{123}^{-1} =  
\tilde  {\cal R}_{12}\ \tilde {\cal R}_{13}\  {\mathfrak F}_{132 }\ {\cal R}_{23}\   {\mathfrak F}_{123}^{-1} = 
\nonumber\\
&&   \tilde  {\cal R}_{12}\ \tilde {\cal R}_{13}\  \tilde {\cal R}_{23}\  {\mathfrak F}_{123 }\ {\mathfrak F}_{123}^{-1} =   
\tilde  {\cal R}_{12}\  \tilde { \cal R}_{13}\  \tilde {\cal R}_{23}.  \nonumber
\end{eqnarray}
Similarly, for the right hand side of (\ref{trans}) we end up with $  \tilde  {\cal R}_{23}\  \tilde { \cal R}_{13}\  \tilde {\cal R}_{12}$, 
and this concludes our proof.
\end{proof}

\begin{pro}{(Drinfeld 3) \label{Drinfeld3}} 
Let  ${\cal R}$ be a solution of the Yang-Baxter equation and  ${\mathfrak F}$ be an admissible twist:
$\tilde{\cal R} ={\mathfrak  F}^{op}\ {\cal R}\ {\mathfrak  F}^{-1}$.
Let also   
${\mathfrak F}_{012...N}$ defined in (\ref{gencoc}), ${\mathfrak T}_{0,12..N} := {\cal R}_{0N} \ldots {\cal R}_{01}$ 
and $\tilde {\mathfrak T}_{0,12...N} := \tilde {\cal R}_{0N} \ldots \tilde {\cal R}_{01}$, then
\begin{equation}
{\mathfrak F}_{12...N0}\ {\mathfrak T}_{0,12..N\ }\ {\mathfrak F}^{-1}_{012..N} = \tilde {\mathfrak T}_{0,12...N} \label{main2}
\end{equation} 
\end{pro}
\begin{proof}
The proof of this Proposition follows exactly the same logic as the proof of Proposition \ref{Drin1}. 
We start with the LHS of (\ref{main2}), and use Proposition \ref{lemma1}:
\begin{eqnarray}
&& {\mathfrak F}_{12...N0}\ {\cal R}_{0N}\  {\cal R}_{0N-1} \ldots {\cal R}_{01}\ {\mathfrak F}^{-1}_{012..N}  = \nonumber\\
&&  \tilde {\cal R}_{0N}\  {\mathfrak F}_{12...N-1 0 N}\ {\cal R}_{0N-1} \ldots {\cal R}_{01}\ {\mathfrak F}^{-1}_{012..N}  = 
\ldots \nonumber\\
&& \tilde {\cal R}_{0N}\  \tilde {\cal R}_{0N-1} \ldots \tilde {\cal R}_{01}\ {\mathfrak F}_{012...N-1 N}\
 {\mathfrak F}^{-1}_{012..N}  \equiv \tilde {\mathfrak T}_{0.12...N}. \nonumber 
\end{eqnarray}
\end{proof}

\subsection{Quasi-admissible twists for set-theoretic solutions}

\noindent We are focusing henceforth on $R$ matrices acting of tensor products of finite ${\cal N}$ dimensional spaces $V,$ 
 i.e. $R: V \otimes V \to V \otimes V.$ 
Our main aim is to identify the explicit form of Drinfeld's twists for finite, 
involutive, non-degenerate, set-theoretic solutions  of the YBE coming from braces,
and also show that they are admissible. We are employing the term {\it quasi-admissible twist} due to the fact
 that we only show here the admissibility of the set-theoretic twists, 
i.e. the validity of the co-cycle condition,
however we are not exploring the action of a co-unit on such twists. 
We make some preliminary remarks at the end of the manuscript illustrating the singular nature of 
set-theoretic twists via a simple example.

More specifically, in the analysis that follows we focus on finite, involutive, non-degenerate, set-theoretic solutions of the 
YBE given by $r ={\cal P} \check r$, where\\ 
$\check r = \sum_{x,y, \in X} e_{x,\sigma_{x}(y)} \otimes e_{y, \tau_{y}(x)}$.
We review below the fundamental constraints emerging from the YBE for such solutions. 
\\

\noindent {\it {\bf Constraints for set-theoretic solutions of the YBE.}} {\it For any finite, non-degenerate, involutive, 
set-theoretic solution of the YBE $ r = {\cal P} \check r$, where\\
$\check r = \sum_{x,y \in X} e_{x, \sigma_x(y)} \otimes e_{y, \tau_y(x)}$, the following conditions hold:
\begin{equation}
 \sigma_{\sigma_{ \eta}(\hat x)}(\sigma_{\tau_{x}(\eta)}(y)) =\sigma_{\eta}(\sigma_{x}({y}))  \quad \mbox{and} \quad 
\tau_{\tau_{y}(x)}(\tau_{\sigma_{x}(y)}(\eta)) = \tau_{y}(\tau_{x}(\eta)), \label{A0}
\end{equation}
$\forall\ \eta, x, y \in X$: $\sigma_{\eta}(x),\ \tau_{y}(x)$ are fixed.}
\begin{proof}
It is convenient for our purposes here to express the Yang-Baxter equation in the following form:
\begin{equation}
\check r_{23}\ T_{1,23}  =   T_{1,23}\ \check r_{23} \label{yb2}
\end{equation}
(similarly $\check r_{12} \ T_{12,3} =   T_{12,3}\ \check r_{12}$), where we define,
$T_{1,23}=  r_{13}\  r_{12}$, 
($T_{12,3} = r_{13}\  r_{23}$) and recall  $\check r = {\cal P} r.$ 

We take the following steps.
\begin{enumerate}
\item  We first  identify $T_{1,23}= r_{13} r_{12}$ , where $r ={\cal P} \check r$  and $\check r$ is the set-theoretic solution
\begin{equation}
T_{1,23} = \sum_{x,y,\eta \in X} e_{\tau_{y}(x), \sigma_{\eta}(x) } \otimes 
e_{\eta, \tau_x(\eta)} \otimes e_{\sigma_x(y), y}. \label{co1}
\end{equation}
\item We calculate the left and right hand side of (\ref{yb2}), using also the involutive property, 
\begin{eqnarray}
LHS & =&  \sum_{x, y,\eta \in X}e_{\tau_{y}(x), \sigma_{\eta}(x)} \otimes 
e_{\sigma_{\eta}(\sigma_x(y)), \tau_x(\eta)} \otimes e_{\tau_{\sigma_x(y)}(\eta), y} 
\label{s1}\\
RHS &=&  \sum_{x, y,\eta \in X}e_{\tau_{y}(x), \sigma_{\eta}(x)}  \otimes 
e_{\eta, \sigma_{\tau_x(\eta)}(y)} \otimes e_{\sigma_x(y), \tau_y(\tau_x(\eta))}. \label{s2}
\end{eqnarray}
\item We identify the constraints emerging from (\ref{yb2}) by equating expressions (\ref{s1}) and  (\ref{s2}). 
Equivalence of expression (\ref{s1}) and (\ref{s2}) leads to
\begin{eqnarray}
&& \eta =   \sigma_{\hat \eta}(\sigma_{\hat x}(\hat y)),  \quad \sigma_{x}(y) = 
\tau_{\sigma_{\hat x}(\hat y)}(\hat \eta)  \quad \mbox{and} \label{cc1}\\
&& \tau_x(\eta) = \sigma_{\tau_{\hat x}(\hat \eta)}(\hat y), \quad y = \tau_{\hat y}(\tau_{\hat x}(\hat \eta)). \label{cc2}
\end{eqnarray}
Note that relations (\ref{cc1}) are equivalent to: $\sigma_{\eta}(\sigma_x(y)) = \hat \eta,$ 
and $\tau_{\sigma_x(y)}(\eta) = \sigma_{\hat x}(\hat y)$. 
We also require
\begin{equation}
\sigma_{\eta}(x) = \sigma_{\hat \eta}(\hat x) \quad \mbox{and}  \quad \tau_{y}(x) = 
\tau_{\hat y}(\hat x). \label{cc3}
\end{equation}

Conditions (\ref{cc1})-(\ref{cc3}) guarantee the equivalence between 
the LHS and RHS of (\ref{yb2}) given in (\ref{s1}), (\ref{s2}).

\item We consider together (\ref{cc2}), (\ref{cc3}) and obtain $\sigma_{\eta}(x) =  
\sigma_{\hat \eta}(\hat x), \ \tau_x(\eta) = \sigma_{\tau_{\hat x}(\hat \eta)}(\hat y)$, 
which leads to $\eta = \sigma_{\sigma_{\hat \eta}(\hat x)}(\sigma_{\tau_{\hat x}(\hat \eta)}(\hat y))$,
then comparing with (\ref{cc1}) we obtain the first fundamental constraint:
\begin{equation}
C_1 \equiv \sigma_{\sigma_{\hat \eta}(\hat x)}(\sigma_{\tau_{\hat x}(\hat \eta)}(\hat y)) - 
\sigma_{\hat \eta}(\sigma_{\hat x}({\hat y})) =0, \label{A}
\end{equation}
and conversely if $C_1 =0$ then $\sigma_{\eta}(x) = \sigma_{\hat \eta}(\hat x)$ given (\ref{cc1}), (\ref{cc2}).

Similarly, from (\ref{cc1}), (\ref{cc3}) $\sigma_x(y) = \tau_{\sigma_{\hat x}(\hat y)}(\hat \eta)$, 
$\tau_{y}(x) = \tau_{\hat y}(\hat x)$, which leads to 
$y = \tau_{\tau_{\hat y}(\hat x)}(\tau_{\sigma_{\hat x}(\hat y)}(\hat \eta))$, 
and comparison with (\ref{cc2}) leads to the second fundamental constraint:
\begin{equation}
C_2 \equiv  \tau_{\tau_{\hat y}(\hat x)}(\tau_{\sigma_{\hat x}(\hat y)}(\hat \eta)) -
\tau_{\hat y}(\tau_{\hat x}(\hat \eta)) =0, \label{B}
\end{equation}
and conversely if $C_2 =0$ then $\tau_{\hat y}(\hat x) =\tau_{y}(x),$ given (\ref{cc1}), (\ref{cc2}).
\end{enumerate}

For the brace solution of the YBE, which is our main interest in this study, 
the above constraints are satisfied as shown by Rump.
\end{proof}
We derive in the following Lemma the explicit forms of the co-products of 
the quantum algebra associated to brace solutions of the YBE. 
This will be of significance when identifying the corresponding admissible Drinfeld twists.
\begin{lemma}{\label{cor0}}
Let  $\check r$ be an involutive, finite, non-degenerate set-theoretic solution of the YBE and 
$T_{1,23} = \sum_{x,y \in X}e_{x,y} \otimes \Delta({\cal L}_{x,y})$, where ${\cal L}_{x,y}$ are 
the represented elements of the corresponding quantum algebra defined by the RTT relation (\ref{RTT}) and
$r = {\cal P} \check r$ $({\cal P}$ is the permutation operator). 
Then 
\begin{equation}
\Delta({\cal L}_{x,y})\ \check r =   \check r\ \Delta({\cal L}_{x,y})  \quad \forall x,\ y \in X, \label{cor1}
\end{equation}
subject to (\ref{A}), (\ref{B}).
\end{lemma}
\begin{proof}
The proof is an immediate consequence of the constraints emerging for set-theoretic solutions of the YBE. 
Indeed, from (\ref{co1}) we can read of 
\begin{equation}
\Delta({\cal L}_{\tau_y(x), \sigma_{\eta}(x)}) =  \sum_{x,y,\eta \in W}
e_{\eta, \tau_x(\eta)} \otimes e_{\sigma_x(y), y},
\end{equation}
where {\it $W: x,y,\eta \in X$ such that $\sigma_{\eta}(x)$ and $\tau_{y}(x)$ are fixed}. \\
Then due to $ \check r_{23}\  T_{1,23}\ =T_{1,23}\ \check r_{23} $ 
we obtain  (\ref{cor1}). i.e. the entries  of the $T_{1,23}$-matrix commute with $\check r$, subject to (\ref{A}), (\ref{B}).
\end{proof}
\noindent It was shown in \cite{DoiSmo2} Proposition 3.3 that any brace solution of the YBE
can be obtained from the permutation operator via a similarity transformation, i.e. 
$\check r_{12} = {\cal F}^{-1}_{12}{\cal P}_{12} {\cal F}_{12}$ 
\footnote{The twist of the present study is the inverse of ${\cal F}$ in 
\cite{DoiSmo2}. It is just a matter of convention}.

\begin{pro} \label{Drintw1} Let $\check r = \sum_{x,y \in X} e_{x, \sigma_x(y)} \otimes e_{y, \tau_y(x)}$
 be the brace solution of the braid YBE.  Let also $V_{k}$, $k \in \{1, \ldots, {\cal N}^2 \}$ be the eigenvectors 
of the permutation operator ${\cal P} = \sum_{x,y\in X} e_{x,y} \otimes e_{y,x}$, and  $\hat V_{k}$, $k \in \{1, \ldots, {\cal N}^2 \}$ 
be the eigenvectors of the brace  $\check r$-matrix.
Then the $\check r$-matrix can be expressed as  a Drinfeld twist, such that 
$\check r={\cal F}^{-1} {\cal P} {\cal F}$, where the twist ${\cal F}^{-1}$ is expressed as 
${\cal F}^{-1} = \sum_{k=1}^{{\cal N}^2} \hat V_k\  V_k^T$  \cite{DoiSmo2}.
\end{pro}
From the proof of the above Proposition (Proposition 3.3 in \cite{DoiSmo2}) we know that the eigenvectors 
of  the permutation operator ${\cal P}$ and the set-theoretic braid solution $\check r$ are given as follows. 
Let $\hat e_k,\  k \in \{ 1, \ldots, {\cal N}\}$ be the ${\cal N}$ 
dimensional column vector with 1 in the $k^{th}$ position and 0 elsewhere, 
i.e. $\hat e_{k}$ form a basis of the ${\cal N}$ dimensional vector space.
\begin{enumerate}
\item The (normalized) eigenvectors of the permutation operator are ($x,\ y \in X$):
\begin{eqnarray}
&& V_{k} =  \hat e_x \otimes \hat e_x, \quad  k \in \big \{1, \ldots,  {\cal N}\big \} \nonumber\\
&& V_{k} = {1\over \sqrt{2}} \big ( \hat e_x \otimes \hat e_y + \hat e_y \otimes \hat e_x \big),    \quad x\neq y,
\quad k \in \big \{{\cal N}+1, \ldots, {{\cal N}^2 + {\cal N} \over 2}\big \}, \nonumber\\
&&  V_{k} = {1\over \sqrt{2}} \big ( \hat e_x \otimes \hat e_y - \hat e_y \otimes \hat e_x \big),    \quad x\neq y,
\quad k \in \big \{ {{\cal N}^2 + {\cal N} \over 2}+1, \ldots, {\cal N }^2 \big \}. \label{0th}
\end{eqnarray}
The first ${{\cal N}^2 + {\cal N} \over 2}$ eigenvectors  have the same eigenvalue $1$, while the
 rest ${{\cal N}^2 - {\cal N} \over 2}$ eigenvectors  have eigenvalue $-1$. Also it is easy to check that $V_k$ 
form an ortho-normal basis for the ${\cal N}^2$ dimensional space. Indeed, 
$V^T_k V_l = \delta_{kl}$ and 
$\sum_{k=1}^{{\cal N}^2} V_{k}V_{k}^T = I_{{\cal N}^2}$ ($^T$ denotes usual transposition).

\item The eigenvectors of the $\check r$-matrix are
\begin{eqnarray}
&& \hat  V_{k} =  \hat e_x \otimes \hat e_y,  \quad (x, y) = (\sigma_{x}(y), \tau_y(x)), \quad  k \in \big \{1, \ldots,  {\cal N}\big\} \nonumber\\
&& \hat V_{k} = {1\over \sqrt{2}} \big ( \hat e_x \otimes \hat e_y + \hat e_{\sigma_x(y)} \otimes \hat e_{\tau_y(x)} \big),   
~~~~~~k \in \big \{{\cal N}+ 1, \ldots, {{\cal N}^2 + {\cal N} \over 2}\big \}, \nonumber\\
&&  \hat V_{k} = {1\over \sqrt{2}} \big ( \hat e_x \otimes \hat e_y - \hat e_{\sigma_x(y)} \otimes \hat e_{\tau_y(x)}\big),   
~~~~(x,y )\neq (\sigma_x(y), \tau_y(x)), \nonumber \\
&& k \in \big \{ {{\cal N}^2 + {\cal N} \over 2}+1, \ldots, {\cal N }^2 \big \} \label{1st}
\end{eqnarray}
As in the case of the permutation operator the $\check r$ matrix has the same eigenvalues
$1$ and $-1$ and the same multiplicities, 
${{\cal N}^2 + {\cal N} \over 2}$ and  ${{\cal N}^2 - {\cal N} \over 2}$ respectively. 
Hence, the two matrices are similar, i.e. there exists some invertible ${\cal F} \in \mbox{End}({\mathbb C}^{\cal N}) \otimes \mbox{End}({\mathbb C}^{\cal N})$ 
(not uniquely defined) such that $\check r = {\cal F}^{-1} {\cal P} {\cal F}$, where ${\cal F}^{-1} = \sum_{k=1}^{{\cal N}^2} \hat V_k\  V_k^T.$
\end{enumerate}

From Proposition 3.3 in \cite{DoiSmo2} we can extract explicit forms for the twist ${\cal F}$ and state the following.

\begin{pro}{\label{twistlocal}} Let ${\cal F}^{-1} = \sum_{k=1}^{{\cal N}^2} \hat V_k\  V_k^T$ 
be the similarity transformation (twist) of Proposition \ref{Drintw1}, such that $\check r = {\cal F}^{-1} {\cal P} {\cal F},$
where $\check r = \sum_{x,y \in X} e_{x, \sigma_x(y)} \otimes e_{y, \tau_y(x) }$ is the brace solution of the braid YBE,  ${\cal P}$ is the permutation 
operator and $\hat V_k,\ V_k$ are their respective eigenvectors.  Then the twist can be explicitly expressed as 
${\cal F} = \sum_{x,y \in X} e_{x,x} \otimes e_{\sigma_{x}(y), y}$. 
\end{pro}
\begin{proof}
We begin our proof by  re-expressing the eigenvectors of the permutation operator in a convenient for our purposes form.
The first ${\cal N}$ eigenvectors are used as they are, we only conveniently re-express the rest  (let $n = {{\cal N}^2 +{\cal N} \over 2}$):
\begin{eqnarray}
&& V_{k} =  \hat e_x \otimes \hat e_x, \quad  k \in \big \{1, \ldots,  {\cal N}\big \} \nonumber\\
&& V_{k} = {1\over \sqrt{2}} \big ( \hat e_x \otimes \hat e_{\sigma_x(y)} + \hat e_{\sigma_x(y)} \otimes \hat e_x \big),    
~~ x\neq \sigma_{x}(y),~~~ k \in \big \{{\cal N}+1, \ldots, n\big \}, \nonumber\\
&&  V_{k} = {1\over \sqrt{2}} \big ( \hat e_x \otimes \hat e_{\sigma_{x}(y)} - \hat e_{\sigma_{x}(y)} \otimes \hat e_x \big),   
~~ x\neq \sigma_{x}(y),
~~ k \in \big \{ n+1, \ldots, {\cal N }^2 \big \}. \label{2nd}
\end{eqnarray}

Let us first mention that if $x = \sigma_x(y)$ then  $y = \tau_{x}(y)$ (and vice versa). 
Indeed, this can be show as follows: 
let $\hat y =\tau_{y}(x),$ then due to involution we obtain $x = \sigma_x(\hat y)$ (and 
$y = \tau_{\hat y}(x)$), but due to the fact that $\sigma_x(y)$ is a bijection, and also 
$x = \sigma_x(y),$ we conclude that $\hat y = y$.  We may now compute ${\cal F}^{-1}$ 
using the explicit expressions of the eigenvectors of $\check r$ and ${\cal P}$  
(\ref{1st}), (\ref{2nd})  (recall $n = {{\cal N}^2 +{\cal N} \over 2}$) $x,\ y \in X$: 
\begin{eqnarray}
{\cal F}^{-1} &=& \sum_{k=1}^{{\cal N}} \hat V_k\  V_k^T  + \sum_{k={\cal N}+1 }^{n} \hat V_k\  V_k^T  
+ \sum_{k=n+1}^{{\cal N}^2} \hat V_k\  V_k^T \nonumber\\
&=& \sum_{x = \sigma_{x}(y)} (\hat e_x \otimes \hat e_y)(\hat e_x^T \otimes \hat e_x^T) \nonumber\\ 
&+& {1\over 2} \sum_{x \neq \sigma_x(y)}  (\hat e_x \otimes \hat e_y)(\hat e_{x}^T \otimes \hat e_{\sigma_x(y)}^T) +  
{1\over 2} \sum_{x \neq \sigma_x(y)}  (\hat e_{\sigma_x(y)} \otimes \hat
 e_{\tau_y(x)})(\hat e_{\sigma_x(y)}^T \otimes \hat e_x^T) \nonumber\\
&=& \sum_{x = \sigma_x(y)} e_{x,x} \otimes e_{y,x} + \sum_{x \neq  \sigma_x(y)} e_{x,x} 
\otimes e_{y,\sigma_{x}(y)}  \nonumber\\ &= & \sum_{x , y \in X} e_{x,x} \otimes e_{y,\sigma_{x}(y)} \nonumber
\end{eqnarray}
From the above expression and due to the fact that $\sigma_x$ and $\tau_y$ are bijections, we conclude that
$ {\cal F}=  \sum_{x,y \in X} e_{x,x} \otimes e_{\sigma_{x}(y), y}$,  
(${\cal F}_{12}^{-1} = {\cal F}_{12}^{T},$ where $^T$ 
denotes total transposition in both spaces).

Recall that  $r = {\cal P} \check r,$ we also confirm by direct computation, 
and using the fact that $\sigma_x,\ \tau_y$ 
are bijections that $({\cal F }^{(op)})^{-1} {\cal F}= 
\sum_{x, \in X} e_{y, \sigma_x(y)} \otimes e_{x, \tau_y(x)}  =r$.
\end{proof}

\begin{rem}{\label{rg}} The twist is not uniquely defined, for instance an alternative twist is of the form 
${\cal G} = \sum_{x,y\in X}e_{\tau_{y}(x), x} \otimes e_{y,y} $, 
and $\sum_{x, \in X} e_{y, \sigma_x(y)} \otimes e_{x, \tau_y(x)}= {\cal G}_{21}^{-1} {\cal G}_{12}$. 
This is  shown by direct computation.
\end{rem}

\begin{rem}{} The Baxterized  solution of the YBE, $R(\lambda) = \lambda   r + {\cal P}$ (recall ${\cal P}$ is the permutation operator) 
can be also expressed as $R_{12}(\lambda)= {\cal F}^{-1}_{21} \tilde R_{12}(\lambda) {\cal F}_{12}$, 
where $\tilde R$ is the Yangian. This is a straightforward statement due to the form of the Yangian 
$\tilde R(\lambda) = \lambda{\mathbb I} + {\cal P}$, 
and due to the fact that ${\cal F}_{21}^{-1} {\cal P}_{12} {\cal F}_{12} = {\cal P}_{12}$.
\end{rem}

We  now focus on the derivation of the 3-twist ${\cal F}_{012}$ and the co-cycle condition. 
We will now identify the 3-twist for any finite, involutive, non-degenerate set-theoretic solution and inspired by 
Propositions (\ref{lemma1})-(\ref{Drinfeld3}) will show that a suitable co-cycle condition  
is satisfied, i.e. ${\cal F}$ is an admissible twist. This is achieved in the following Propositions.

\begin{pro}{\label{proco1}} In accordance to the co-multiplication of the quantum algebra as defined in Lemma \ref{cor0}, 
we define:
\begin{eqnarray}
&& {\cal F}_{0,12} = \sum_{x, y, \eta \in X} e_{\sigma_{\eta}(x), \sigma_{\eta}(x)} \otimes e_{\eta, \tau_x(\eta)}\
 \otimes e_{\sigma_{x}(y), y} \vert_{C_1=0} \label{cop1}\\
&& {\cal F}_{01,2} =\sum_{x, y, \eta \in X} e_{\sigma_{\eta}(x), \sigma_{\eta}(x)} \otimes e_{\tau_x(\eta), \tau_x(\eta) }\ 
\otimes e_{\sigma_{\eta}(\sigma_{x}(y)), y} \vert_{C_1=0} \label{cop2}
\end{eqnarray}
where $C_1 = \sigma_{\sigma_{\eta}(x)}(\sigma_{\tau_{x}(y)}(y)) - \sigma_{\eta}(\sigma_x(y))$. 
Let also $\check r = \sum_{x, y \in X } e_{x, \sigma_{x}(y)} \otimes e_{y, \tau_{y}(x)}$, then 
\begin{equation}
 \check r_{01} {\cal F}_{01,2} = {\cal F}_{01,2} \check r_{01},\  
\quad  \check r_{12} {\cal F}_{0,12} = {\cal F}_{0,12}\check r_{12}. \label{comut1}
\end{equation}
\end{pro}
\begin{proof}
We divide the proof in two parts.
\begin{enumerate}
\item We first show that $ \check r_{01} {\cal F}_{01,2} = {\cal F}_{01,2} \check r_{01}$, using the 
explicit expressions of $\check r$, ${\cal F}_{0,12}$ and ${\cal F}_{01,2}$: 
The LHS of the expression is (subject to the constraint $C_1 =0$)
\begin{eqnarray}
&& \check r_{01} {\cal F}_{01,2} = 
\sum_{\eta, x,y \in X} e_{\eta, \sigma_{\eta}(x)} \otimes e_{x, \tau_{x}(\eta)} \otimes 
e_{\sigma_{\eta}(\sigma_x(y)), y}. \nonumber
\end{eqnarray}
Similarly,
\begin{eqnarray}
&& {\cal F}_{01,2} \check r_{01} = 
\sum_{\eta, x,y \in X} e_{\eta, \sigma_{\eta}(x)} \otimes 
e_{x, \tau_{x}(\eta)} \otimes e_{\sigma_{\sigma_{\eta}(x)}(\sigma_{\tau_x(\eta)}(y)), y}, \nonumber
\end{eqnarray}
and due to the constraint $C_1=0$:
${\cal F}_{01,2} \check r_{01} = \sum_{\eta, x,y \in X} e_{\eta, \sigma_{\eta}(x)} 
\otimes e_{x, \tau_{x}(\eta)} \otimes e_{\sigma_{\eta}(\sigma_x(y)), y}$, 
which concludes the first part of our proof. 

\item We move on now to the second part of our proof, which is a bit more involved: 
$ \check r_{12} {\cal F}_{0,12} = {\cal F}_{0,12} \check r_{12}$.
By direct computation given the explicit expressions for ${\cal F}_{0,12}$ and $\check r_{12}$:
\begin{eqnarray}
&& \check r_{12} {\cal F}_{0,12} = \sum_{\eta, x, y \in X} e_{\sigma_{\eta}(x), \sigma_{\eta}(x)} 
\otimes e_{\sigma_{\eta}(\sigma_x(y)), \tau_x(\eta)}\otimes
e_{\tau_{\sigma_{x}(y)}(\eta), y} \label{expb1} \nonumber\\
&&  {\cal F}_{0,12} \check r_{12} = \sum_{\eta, x,y\in X}  e_{\sigma_{\eta}(x), \sigma_{\eta}(x)} 
\otimes e_{h, \sigma_{\tau_x(\eta)}(y)} \otimes e_{\sigma_x(y), \tau_y(\tau_x(\eta))}, \label{expb2} \nonumber
\end{eqnarray}
subject to $C_1=0$. We now follow the logic of obtaining the constraints for set-theoretic solutions of the YBE. 
We employ (\ref{cc1}) and (\ref{cc2})
which as argued due to $C_1=0 $ guarantee that $\sigma_{\eta}(x) = \sigma_{\hat \eta}(\hat x)$
leading to $ \check r_{12} {\cal F}_{0,12} = {\cal F}_{0,12} \check r_{12} $, which basically concludes our proof. 
\end{enumerate}
All possible permutations among the indices  $0, 1, 2$ 
can be considered in  a straightforward manner. For instance, by multiplying  
expression $ \check r_{12} {\cal F}_{0,12} = {\cal F}_{0,12}\check r_{12}$ 
with ${\cal P}_{01}$ (left and right) we obtain $  \check r_{02} {\cal F}_{1,02} 
= {\cal F}_{1,02}\check r_{02}$.
\end{proof}

\begin{rem}{\label{rem4}}This is a straightforward, but useful remark. Recall, 
$r = {\cal P} \check r$, where ${\cal P}$ is the permutation operator.
If ${\cal F}_{01,2} \check r_{01} = \check r_{01}{\cal F}_{01,2},$ and 
${\cal F}_{1,02} \check r_{02} = \check r_{02}{\cal F}_{1,02},$ as in Proposition \ref{proco1},
then by multiplying the latter two equalities with ${\cal P}$ from the left we conclude: 
${\cal F}_{10,2}  r_{01} = r_{01}{\cal F}_{01,2},$  and
${\cal F}_{1,20} r_{02} =  r_{02}{\cal F}_{1,02}.$
\end{rem}

\noindent {\bf Generalization.} The $(n+1)$-objects  ${\cal F}_{0,12...n}$
and ${\cal F}_{0,12...n}$ 
can be now derived by iteration exploiting the explicit form of 
$T_{0,12..n} =r_{0n}\ldots r_{01}$. Indeed, let us first identify $T_{0,1..n}$ 
recalling the specific form of the set-theoretic solution 
$r =\sum_{x,y \in X} e_{y, \sigma_x(y)} \otimes e_{x, \tau_y(x)},$
\begin{equation}
T_{0,1..n}= \sum_{x_1,..x_n, y_1 \in X} e_{y_n, \sigma_{x_1}(y_1)} \otimes e_{x_1, 
\tau_{y_1}(x_1)}\otimes \ldots \otimes e_{x_n, \tau_{y_n}(x_n)}, \label{genT}
\end{equation}
subject to $~y_{m-1} = \sigma_{x_m}(y_m),$ $m \in \{2,\ldots, n \}$. 
Then  according to the generic expression (\ref{genT}), (bear  in mind that 
$T_{0,1..n} \check r_{mm+1}= \check r_{m m+1} T_{0,1..n},$ 
$m \in \{1,\ldots, n-1\};$ see also Proposition  \ref{proco1} and expression (\ref{cop1})) we derive,\\ 
${\cal F}_{0,1...n} = \sum_{x_1,...,x_n, y_n \in X} e_{\sigma_{x_1}(y_1), \sigma_{x_1}(y_1)} \otimes e_{x_1, 
\tau_{y_1}(x_1)}\otimes \ldots \otimes e_{x_n, \tau_{y_n}(x_n)},$ subject to $~y_{m-1} 
= \sigma_{x_m}(y_m),$ $m \in \{2,\ldots, n \}$. 
Also, by iteration (see expression (\ref{cop2})) we obtain, 
${\cal F}_{01..n-1,n} =\sum_{x_1,...,x_{n+1} \in X} e_{x_1, x_1} \otimes e_{x_2, x_2} 
\otimes \ldots e_{x_n, x_n} \otimes e_{{\mathrm X}, x_{n+1}},$
where we define ${\mathrm X} := \sigma_{x_1}(\sigma_{x_2}(...(\sigma_{x_n}(x_{n+1}))...) $. 
The general expressions hold subject to the constraint 
$\sigma_{x_n}(\sigma_{x_{n+1}}(y)) =\sigma_{\sigma_{x_n}(x_{n+1})}(\sigma_{\tau_{x_{n+1}}(x_n)}(y)) ,$ $\ x_n, y \in X$.

We may now proceed in proving the admissibility of the derived twist, i.e. show
the validity of the co-cycle condition.

\begin{pro}{\label{cocycle}}
 Let ${\cal F}_{01} ={\cal F}  \otimes I$ and ${\cal F}_{12} =I \otimes {\cal F}$, where ${\cal F} =
\sum_{\eta, x,y \in X} e_{\eta, \eta} \otimes e_{\sigma_{\eta}(x), x}$.
Let also  ${\cal F}_{01,2}$ and ${\cal F}_{0,12}$ defined in (\ref{cop1}) and (\ref{cop2}). Then 
\begin{equation}
{\cal F}_{012}:={\cal F}_{01} {\cal F}_{01,2} = {\cal F}_{12} {\cal F}_{0,12}.
\end{equation}
\end{pro}
\begin{proof}
By substituting the expressions for ${\cal F}_{01},\ {\cal F}_{12}$, ${\cal F}_{01,2}$ and ${\cal F}_{0,12}$ 
(recall $C_1=0$ holds for ${\cal F}_{01,2},$ ${\cal F}_{0,12}$) we obtain by direct computation:
\begin{eqnarray}
&&{\cal F}_{01} {\cal F}_{01,2} = 
  \sum_{\eta, x, y \in X}  e_{\sigma_{\eta}(x),\sigma_{\eta}(x)} \otimes e_{\eta, \tau_{x}(\eta)} \otimes 
e_{\sigma_{\eta}(\sigma_x(y)), y}. \nonumber
\end{eqnarray}
Similarly,
\begin{eqnarray}
&& {\cal F}_{12} {\cal F}_{0,12} = 
 \sum_{\eta,x,y \in X}  e_{\sigma_{\eta}(x), \sigma_{\eta}(x)} \otimes e_{\eta, \tau_x(\eta)} 
\otimes  e_{\sigma_{\eta}(\sigma_{x}(y)), y} \nonumber
\end{eqnarray} 
The explicit form the 3-twist is given from the expressions above as\\
${\cal F}_{012}= \sum_{\eta,x,y \in X}  
e_{\sigma_{\eta}(x), \sigma_{\eta}(x)} \otimes e_{\eta, \tau_x(\eta)} \otimes  
e_{\sigma_{\eta}(\sigma_{x}(y)),y}\vert_{C_1=0}.$
\end{proof}

We are now in the position to show the factorization of the monodromy matrix 
$T_{0,12}$ in terms of the admissible twists.

\begin{cor}{\label{twistglobal}} 
Let $r =\sum_{x,y \in X} e_{y, \sigma_x(y)} \otimes e_{x, \tau_y(x)}$, a solution of the YBE,  and 
${\cal F}_{012} =   \sum_{x, y, \eta \in X} e_{\eta, \eta} \otimes e_{\sigma_{\eta}(x), x}\otimes 
e_{\sigma_{\eta}(\sigma_{x}(y)), y}\vert_{C_1=0},$ as defined in Proposition \ref{cocycle}.
Let also, $T_{0,12} = r_{02} r_{01}$, then $T_{0,12} = {\cal F}^{-1}_{120}\ {\cal F}_{012}$ .
\end{cor}
\begin{proof} This is based on Proposition \ref{Drinfeld3}.
Indeed, to show the decomposition of the monodromy matrix $T_{0,12}$ constructed from set-theoretic solutions
we employ the following statements: 
\begin{enumerate}
\item  Proposition \ref{Drinfeld3}: ${\cal F}_{n0}^{-1}\  {\cal F}_{0n}= r_{0n}$, $n \in \{1,\ 2\}$.
\item  Proposition \ref{proco1} and Remark \ref{rem4}: ${\cal F}_{n0,m} r_{0n} = r_{0n} {\cal F}_{0n,m}$ and \\
${\cal F}_{m, n0} r_{0n} = r_{0n} {\cal F}_{m,0n},$  $n\neq m \in \{ 1,\ 2\}$.
\item Proposition \ref{cocycle}: the co-cycle condition,  ${\cal F}_{01} {\cal F}_{01,2} = {\cal F}_{12} {\cal F}_{0,12}$.
\end{enumerate}
We first recall that ${\cal F}_{n0,m} ={\cal P}_{0n} {\cal F}_{0n,m} {\cal P}_{0n},$ 
similarly for ${\cal F}_{m,n0}$, where ${\cal P}$ is the permutation operator.
The proof is straightforward based on the proof of Proposition \ref{Drinfeld3}: i.e. use (1)-(3) above, then 
${\cal F}_{120}\ T_{0,12}\ {\cal F}_{012}^{-1}=I_{V^{\otimes 3}} $.
\end{proof}
 
\noindent {\bf Direct computation.} Let us also confirm Corollary \ref{twistglobal} by direct computation. 
We first derive $T_{0,12}$ via (\ref{genT}) for $n=2$:
\begin{eqnarray}
T_{0,12} =\sum_{\eta,  x, y \in X } e_{y, \sigma_{\eta}(\sigma_x(y))} 
\otimes e_{\eta, \tau_{\sigma_x(y)}(\eta)} \otimes e_{x, \tau_{y}(x)}. \label{TT}
\end{eqnarray}

We also derive ${\cal F}_{012}^{-1}$; due to the fact that $\sigma_y,\  \tau_x$ 
are bijective functions: ${\cal F}_{012}^{-1} = {\cal F}_{012}^{T}$, 
where $^T$ denotes total transposition, i.e. transposition in all three spaces of the tensor product. We may now show
by direct computation that $T_{0,12} = {\cal F}_{120}^{-1}\ {\cal F}_{012}$, indeed
\begin{eqnarray} 
&&{\cal F}_{120}^{-1}\ {\cal F}_{012}= \nonumber\\ 
&&\Big (\sum_{\eta, x, y\in X} e_{y, \sigma_{\eta}(\sigma_x(y))} \otimes e_{\eta, \eta} \otimes e_{x, \sigma_{\eta}(x)}\Big ) 
\Big ( \sum_{\hat \eta, \hat x, \hat y\in X}   e_{\hat \eta, \hat \eta} \otimes e_{\sigma_{\hat \eta}(\hat x), \hat x} 
\otimes e_{\sigma_{\hat \eta}(\sigma_{\hat x}(\hat y)), \hat y}\Big ) = \nonumber\\
&& \sum_{x,y, \eta, \hat x, \hat y,\hat \eta\in X} e_{y, \hat \eta} \otimes e_{\eta, \hat x} \otimes e_{x, \hat y} \label{Fexp1}
\end{eqnarray}
subject to the following conditions:
\begin{equation}
\hat \eta = \sigma_{\eta}(\sigma_x(y)), \quad  \sigma_{\hat \eta}(\hat x) = \eta , \quad 
\sigma_{\eta}(x) = \sigma_{\hat \eta}(\sigma_{\hat x}(\hat y)) \label{conditions1}
\end{equation}
Our aim now is to express the $\hat x, \hat y,\hat \eta$ in terms of $x,y, \eta$,
notice that $\hat \eta$ is already expressed in such a way. Consider now the condition 
$\sigma_{\hat \eta}(\hat x) = h$, and let $\tau_{\hat x}(\hat \eta) = \xi$ then we obtain via the involutive property:
$\hat \eta = \sigma_{\eta}(\xi)$ and $\hat x = \tau_{\xi}(\eta)$. However, due to the first
of the conditions (\ref{conditions1}) and the fact that $\sigma_{\eta}$ is a bijective function, we conclude that 
$\xi = \sigma_{x}(y)$ and hence $\hat x = \tau_{\sigma_x(y)}(\eta) $.
 It remains now to express $\hat y$ in terms of $x, y, \eta$; we consider the third of the conditions 
(\ref{conditions1}) as well as condition $C_1 =0$ then
$\sigma_{\eta}(x) = \sigma_{\sigma_{\hat \eta} (\hat x)}(\sigma_{\tau_{\hat x}(\hat \eta)}(\hat y))$, 
but as shown above $\sigma_{\hat \eta}(\hat x) = h$,
then using also the fact that $\sigma_{\eta}$ is a bijection we conclude
 $x =\sigma_{\tau_{\hat x}(\hat \eta)}(\hat y) $. From our considerations above 
$\tau_{\hat x}(\hat \eta)(=\xi)= \sigma_{x}(y),$ 
and we conclude that $\hat y = \tau_y(x)$. 

Having expressed  $\hat x, \hat y,\hat \eta$ in terms of $x,y, \eta$: $\hat \eta = \sigma_{\eta}(\sigma_x(y)),$ 
$\hat x = \tau_{\sigma_x(y)}(\eta),$  and $\hat y = \tau_y(x),$ we arrive via (\ref{Fexp1}) at
\begin{equation}
{\cal F}_{120}^{-1}\ {\cal F}_{012}=  \sum_{x,y, \eta \in X} e_{y, \sigma_{\eta}(\sigma_x(y))} 
\otimes e_{\eta, \tau_{\sigma_x(y)}(\eta)} \otimes e_{x,\tau_y(x) }  \label{exp2}
\end{equation}
which is precisely $T_{0,12}$ (\ref{TT}). $\hspace{3.1in} \square$

\begin{rem} An alternative admissible twist is derived as follows.
Using  the notation introduced in (\ref{codef}), we define
\begin{eqnarray}
&& {\cal G}_{0,12} = \sum_{\eta, x, y \in X}e_{\tau_y(\tau_x(\eta)), \eta} 
\otimes e_{\sigma_{x}(y), \sigma_{x}(y)} \otimes e_{\tau_y(x), \tau_y(x)} \vert_{C_2=0}\\
&& {\cal G}_{01,2} =  \sum_{\eta, x, y \in X} e_{\tau_x(\eta), \eta} \otimes
 e_{y, \sigma_x(y)} \otimes e_{\tau_y(x), \tau_y(x)}  \vert_{C_2=0}
\end{eqnarray}
$C_2 = \tau_{y}(\tau_{x}(\eta)) - \tau_{\tau_y(x)} (\tau_{\sigma_{x}(y)}(\eta))$, then:

\begin{enumerate}
\item ${\cal G}_{01,2} \check r_{01} =  \check r_{01}{\cal G}_{01,2} $ 
and ${\cal G}_{0,12} \check r_{12} =  \check r_{12}{\cal G}_{0,12} $, 
where $\check r$ is the brace solution.

\item The co-cycle condition is also satisfied and the 3-twist is then derived, i.e.
${\cal G}_{012} :={\cal G}_{01} {\cal G}_{01,2}= {\cal G}_{12} {\cal G}_{0,12}$, 
where the 2-twist is ${\cal G} = \sum_{x,y\in X}e_{\tau_{y}(x), x} \otimes e_{y,y},$ given in Remark \ref{rg}.  
The explicit form of the 3-twist then is\\ ${\cal G}_{012} = 
\sum_{\eta, x, y \in X} e_{\tau_{y}(\tau_x(\eta)), \eta} \otimes e_{y, \sigma_{x}(y)} 
\otimes e_{\tau_{y}(x), \tau_{y}(x)}\vert_{C_2=0}$
\end{enumerate}

The proofs of the statements above follow the same logic of 
the corresponding proofs for  ${\cal F}_{012}$.
\end{rem}

\noindent {\bf The $(n+1)$-twist}. The general $(n+1)$-twists may be derived by iteration using the algebra  co-product rules 
as identified by $T_{0,12....n}= r_{0n} ... r_{02} r_{01}$. For instance, recall the  generalized expressions 
${\cal F}_{0,1...n} = \sum_{x_1,...,x_n, y_n \in X} e_{\sigma_{x_1}(y_1), \sigma_{x_1}(y_1)} \otimes e_{x_1, 
\tau_{y_1}(x_1)}\otimes \ldots \otimes e_{x_n, \tau_{y_n}(x_n)},$ subject to $~y_{m-1} 
= \sigma_{x_m}(y_m),$ $m \in \{2,\ldots, n \},$ and
${\cal F}_{01..n-1,n} =\sum_{x_1,...,x_{n+1} \in X} e_{x_1, x_1} \otimes e_{x_2, x_2} 
\otimes \ldots e_{x_n, x_n} \otimes e_{{\mathrm X}, x_{n+1}},$
where ${\mathrm X} := \sigma_{x_1}(\sigma_{x_2}(...(\sigma_{x_n}(x_{n+1}))...).$ 
We extract the $(n+1)$-twist via the generalized co-cycle condition (\ref{gencoc})  by iteration:
${\cal F}_{012...n}= \sum_{x_1, ...,x_{n+1} \in X} e_{x_1, x_1} \otimes e_{\sigma_{x_1}(x_2), x_2} 
\otimes e_{\sigma_{x_1}(\sigma_{x_2}(x_3)), x_3} \otimes
 \ldots \otimes e_{{\mathrm X}, x_{n+1}}$, subject to the standard constraint 
$\sigma_{x_n}(\sigma_{x_{n+1}}(y)) =\sigma_{\sigma_{x_n}(x_{n+1})}(\sigma_{\tau_{x_{n+1}}(x_n)}(y)) ,$ $\ x_n,  y \in X$.

\begin{lemma}{} Let ${\cal F}$ be an admissible twist, such that 
$ r_{12} = {\cal F}^{-1}_{21} {\cal F}_{12},$ satisfies the YBE. Let also the Baxterized solutions of the YBE
$R(\lambda) = \lambda r +{\cal P},$ $\tilde R(\lambda) = \lambda {\mathbb I} +{\cal P}$ 
and ${\mathrm T}_0(\lambda)= R_{0N}(\lambda) \ldots R_{01}(\lambda),$  
$~\tilde {\mathrm T}_0(\lambda)= \tilde R_{0N}(\lambda) \ldots \tilde R_{01}(\lambda),$ then
\begin{equation}
 \tilde {\mathrm T}_0(\lambda)= {\cal F}_{12..n0}\ {\mathrm T}_0(\lambda)\  {\cal F}^{-1}_{012..n}. \label{concl}
\end{equation}
\begin{proof}
It suffices to prove that  
\begin{equation}
{\cal F}_{12...n-1 n 0 n+1...N}\ R_{0n}(\lambda) =  \tilde R_{0n}(\lambda)\ 
{\cal F}_{12...n-1 0n n+1...N}  \label{state1}
\end{equation}
then according to Propositions \ref{lemma1} and \ref{Drinfeld3} expression (\ref{concl}) follows. 
Indeed, for the admissible twist via Proposition \ref{lemma1}, $ {\cal F}_{12...n-1 n 0 n+1...N}\ 
r_{0n} =  {\cal F}_{12...n-1 0n n+1...N} ,$ 
but also by the definition of the permutation operator ${\cal P}$ we  have \\
${\cal F}_{12...n-1 n 0 n+1...N}\ {\cal P}_{0n}=  {\cal P}_{0n}\ {\cal F}_{12...n-1 0n n+1...N}.$ 
Then according to the definition of $R,\ \tilde R$ we arrive at (\ref{state1}), which concludes our proof.
\end{proof}
\end{lemma}

\subsection{Special case: Lyubashenko's solution}

\noindent We focus now on a special class of set-theoretic solutions of the YBE known as Lyubashenko's solutions
\cite{Drin} (see also \cite{gateva} in relation to symmetric groups).
We first introduce this class of solutions and  we show that they
can be expressed as simple twists  recalling the results of \cite{DoiSmo2}.
We then move on to show that these are admissible twists and we explicitly derive the associated $n$-twists.

Let us recall  Proposition 3.1  in \cite{DoiSmo2}:
\begin{pro}{\label{prop1L}} Let $\tau,\ \sigma: X \to X$ be isomorphisms, 
such that $\sigma(\tau(x)) = \tau(\sigma(x)) = x$ and let
${\cal V}=\sum_{x \in X} e_{x, \tau(x)}$ and ${\cal V}^{-1} = \sum_{x \in X}  e_{ \tau(x),x}$.
Then any solution of the braid YBE of the type 
\begin{equation}
\check r=\sum_{x, y \in X} e_{x, \sigma(y)} \otimes e_{y, \tau(x)}, \label{special1}
\end{equation}
can be obtained from the permutation operator ${\cal P}= \sum_{x, y \in X} e_{x,y} \otimes e_{y,x}$ as
\begin{equation}
\check r = ({\cal V}\otimes I ){\cal P} ({\cal V}^{-1} \otimes I)= ( I \otimes {\cal V}^{-1} ) {\cal P} (I \otimes  {\cal V} ) \label{special1b}
\end{equation}
\end{pro}

We will now derive the $n$-twists associated to Lyubashenko's solution and 
show that these are admissible.
We introduce in what follows two distinct twists ${\cal F}$ and ${\cal G}$ 
compatible with the results of \cite{DoiSmo2}.
Let us first introduce the twist 
${\cal F}_{12}={\cal V}_2,$
given the explicit form ${\cal V} = \sum_{x \in X} e_{x, \tau(x)}$ we may rewrite
 ${\cal F} = \sum_{x, y \in X}e_{x, x} \otimes e_{y, \tau(y)}$.
It is instructive to compare with the general set-theoretic twist,  Proposition \ref{twistlocal}. 
Note that the terms at each space now ``decouple''  due to the fact that $\sigma_{x}(y),\ \tau_{y}(x) \to \sigma(y),\ \tau(x)$.

Indeed, according to Proposition \ref{prop1L},  ${\cal P} = {\cal V}_2 \check r_{12} {\cal V}_2^{-1}$, 
where $\check r = \sum_{x,y\in X} e_{x,\sigma(y)} \otimes e_{y,\tau(x)}$. Also,
${\cal F}$ is a {\it Reshetikhin type} twist as it trivially satisfies the YBE,
and consequently the co-product structure is the one of the quasi-triangular Hopf algebra:
\begin{equation}
{\cal F}_{1,23} = {\cal F}_{13} {\cal F}_{12}, \quad {\cal F}_{12,3}= {\cal F}_{13} {\cal F}_{23}. \label{cop1b}
\end{equation}
By means of (\ref{cop1b}) and given that ${\cal F}_{12}={\cal V}_2$:
we have: 
\begin{equation}
{\cal F}_{1,23} = {\cal V}_2 {\cal V}_3, \quad {\cal F}_{12,3} = {\cal V}_3^2. \label{cop1bc}
\end{equation} 
The co-cycle condition is satisfied and it is nothing but the Yang-Baxter equation: 
${\cal F}_{123}:= {\cal F}_{12} {\cal F}_{12,3}$.
We also derive by iteration:
\begin{equation}
{\cal F}_{0, 12...n} = {\cal F}_{0n} \ldots {\cal F}_{01}, \quad {\cal F}_{12.n,0} =
{\cal F}_{10}{\cal F}_{20} \ldots  {\cal F}_{n0}, \label{cop2b}
\end{equation}
and the explicit expressions are given as
\begin{equation}
{\cal F}_{0,12..n} = {\cal V} _1 {\cal V}_2 \ldots {\cal V}_n, \quad {\cal F}_{12..n,0} = 
{\cal V}_0^{n}. \label{fcop}
\end{equation}

The precise form of the $(n+1)$-twist is derived in the following Lemma. 
\begin{lemma}{\label{lemmab1}}
Let ${\cal F}_{12} = {\cal V}_2$, be the admissible twist for Lyubashenko's solution, then
 the $(n+1)$-twist is given as $\ {\cal F}_{012..n} =\prod_{k=0}^n {\cal V}_k^k.$
\end{lemma}
\begin{proof}
The proof is straightforward given the explicit forms of ${\cal F},$ 
expressions (\ref{fcop}) and the generalized cocycle condition (\ref{gencoc}).
\end{proof}
Recall also that ${\cal V} = \sum_{x\in X} e_{x, \tau(x)},$ then we can express the $(n+1)$-twist as
\begin{equation}
{\cal F}_{01...n} = \sum_{x_1,...,x_{n+1} \in X} e_{x_1, x_1}
\otimes e_{\sigma(x_2), x_2} \otimes e_{\sigma^2(x_3), x_3} \otimes \ldots 
\otimes  e_{\sigma^n(x_{n+1}), x_{n+1}}. \nonumber
\end{equation}
We shall also need for the next Lemma the expression 
${\cal F}_{12...n0} ={\cal V}_0^{n}{\cal V}_{n}^{n-1} \ldots {\cal V}_3^{2} {\cal V}_2.$

Given the explicit form of the $(n+1)$-twist from Lemma \ref{lemmab1} we can show 
the factorization of the monodromy matrix in a straightforward manner in the following Lemma.
\begin{lemma}{}Let $T_{0,12..N} = r_{0N} ... r_{01},$
where $r$ is Lyubashenko's solution and ${\cal F}_{012..N}$, 
the $(N+1)$-twist as derived in Lemma \ref{lemmab1}, 
then $\ {\cal F}^{-1}_{12..N0}\ {\cal F}_{012...N} = T_{0,12...N}.$
\end{lemma}
\begin{proof}
The factorization can be checked by direct computation (recall in this case, $r_{0n} ={\cal V}_0^{-1} {\cal V}_n$), indeed
\begin{eqnarray}
&& {\cal F}_{12..N0}\ T_{0,12..N}\ {\cal F}^{-1}_{012...N} =\nonumber\\
& & \Big ( {\cal V}_0^{N}{\cal V}_{N}^{N-1} \ldots {\cal V}_3^{2} {\cal V}_2\Big )\Big ({\cal V}_{0}^{-1} 
{\cal V}_{N} \ldots {\cal V}_0^{-1} {\cal V}_{1} \Big ) \Big ({\cal V}_N^{N}{\cal V}_{N-1}^{N-1} 
\ldots {\cal V}_2^{2} {\cal V}_1 \Big )^{-1} =  I^{\otimes (N+1)}. \nonumber
\end{eqnarray}
\end{proof}

\begin{rem}
There is an alternative admissible twist for Lyubashenko's solutions. 
Indeed, ${\cal G}_{12} ={\cal V}_1^{-1},$ is an alternative 
admissible, Reshetikhin type twist:
\begin{enumerate}
\item ${\cal G}_{12,3} = {\cal G}_{13} {\cal G}_{23}$ and ${\cal G}_{1,23} = {\cal G}_{13}{\cal G}_{12}$. 
\item ${\cal G}_{0,12...n} ={\cal V}_0^{-n}$ and ${\cal G}_{12..n,0} = {\cal V}^{-1}_1 {\cal V}^{-1}_2 \ldots {\cal V}_n^{-1} .$ 
\item Via the generalized co-cycle condition the alternative $(n+1)$-twist is expressed as 
${\cal G}_{01...n} =\prod_{k=0}^{n} {\cal V}_k^{-(n-k)} $.
\end{enumerate}
\end{rem}
Recalling also that ${\cal V}^{-1} = \sum_{x\in X} e_{x, \sigma(x)}$ we can write
\begin{equation}
{\cal G}_{01...n} = \sum_{x_1,...,x_{n+1} \in X} e_{\tau^n(x_1),x_1} \otimes e_{\tau^{n-1}(x_2), x_2} \otimes \ldots  \otimes 
e_{\tau(x_n), x_n}\otimes e_{x_{n+1}, x_{n+1}}. \nonumber
\end{equation}

\begin{rem}
The set-theoretic solution (\ref{brace1}) is a representation of the Hecke algebra $H_N(q=1)$ \cite{DoiSmo1}, 
hence the Baxterized solution can be expressed as in Remark 3.12. A $q$-analogue of Lyubashenko’s solution was derived 
in \cite{DoiSmo2}  via the  simple twist ${\cal V}= \sum_{x\in X} e_{\sigma(x), x}$.  Indeed, if $g$ is the ${\mathfrak U}_{q}(\mathfrak{gl}_{\cal N})$ invariant 
representation of the Hekce algebra $H_{N}(q)$ \cite{Jimbo, Jimbo2} then so $g_{t} = ({\cal V}\otimes I) g ({\cal V}^{-1}\otimes I)$ is 
 provided that $sgn(x-y) = sgn(\tau(x) - \tau(y)) =
 sgn(\sigma(x) - \sigma(y))$ $\forall x,y \in X$ (note also that $[{\cal V} \otimes {\cal V}, g_t] =0$). Then the Baxterized solution 
$\check R_t= e^{\lambda} g_t– e^{-\lambda}g_t^{-1}$ gives a twisted $\check R$-matrix associated to the 
${\mathfrak U}_{q}(\mathfrak{gl}_{\cal N})$. The more general twist of Proposition 3.10  can be also applied to 
the ${\mathfrak U}_{q}(\mathfrak{gl}_{\cal N})$  invariant  solution $g$  of the braid relation, subject to certain 
constraints, however this case will be discussed in detail elsewhere.
\end{rem}

\subsection*{Curious observations.} 
This is a preliminary discussion motivating the next natural steps of our investigation.
Throughout  this manuscript we have been focused on the  identification of set-theoretic 
twists  and the issue of their admissibility in the sense of Proposition \ref{cocycle}.
We have restricted our attention on the co-product structure of the underlying quantum algebra and 
we have not discussed the actions of  the co-unit and antipode for set-theoretic solutions, which
are crucial  in identifying the quantum algebra as a quasi-triangular Hopf algebra.
 
Let us  briefly refer to the notions of co-unit and antipode for the quantum algebra 
associated to Lyubashenko's solution, which represents a simple, but characteristic 
example  of involutive, set-theoretic solution.  Recall the Lyubashenko solution can 
be expressed in the compact way $r: V \otimes V \to V \otimes V$ ($V$ is the ${\cal N}$-dim vector space), such that
$r = {\cal V}^{-1} \otimes {\cal V}$, where ${\cal V} = \sum_{x\in X} e_{x, \tau(x)}$ is a group like element. 
From the definition $ T_{0,12} = r_{02} r_{01}$ and recalling the simple form of Lyubashenko's solution 
we have $\Delta({\cal V}) = {\cal V}\otimes {\cal V},$ indeed ${\cal V}$ is a group like element.  We define the co-unit:
$ \epsilon({\cal V})=  \epsilon({\cal V})^{-1}=\epsilon(I) = 1,$ 
then it follows that
$(\epsilon \otimes \mbox{id})\Delta({\cal V})=(\mbox{id}\otimes \epsilon)\Delta({\cal V})={\cal V}.$ We also define 
the antipode: $s({\cal V}) ={\cal  V}^{-1},$ then $m \big ( (s\otimes \mbox{id})\Delta({\cal V})\big )=
m \big ( ( \mbox{id} \otimes s)\Delta({\cal V})\big ) = I_V.$ 

We recall the Lyubashenko twists:  ${\cal F} = I \otimes {\cal V}$ and ${\cal G} = {\cal V}^{-1} \otimes I$,  
we can then readily check, regardless the above defined action of the co-unit, that 
$(\epsilon \otimes\mbox{id}) {\cal F} = {\cal V},$ 
$\ (\mbox{id} \otimes \epsilon) {\cal G} = {\cal V}^{-1}$ and 
via $\epsilon({\cal V})=  \epsilon({\cal V}^{-1})=1,$  we conclude $(\mbox{id} \otimes \epsilon) {\cal F} = I,$ $\ (\epsilon \otimes\mbox{id}) {\cal G} = I.$ 
Even though the conditions $(\epsilon \otimes\mbox{id}) {\cal F} = I,$ $(\mbox{id} \otimes \epsilon){\cal G} =I$ 
are now relaxed,  the twists ${\cal F},\ {\cal G}$  are still  admissible in the weaker sense of  Proposition \ref{cocycle} and Corollary \ref{twistglobal}.
It is also interesting to  present the action of the above defined co-unit on Lyubashenkso's $r$ matrix: $(\epsilon \otimes\mbox{id})r = {\cal V},$ 
$\ (\mbox{id} \otimes \epsilon)r = {\cal V}^{-1}$.
This is an ``uncommon’’ action of the co-unit, the origin of which is the fact that both the set-theoretic 
$r$-matrix as well as the related twists have no semi classical analogues, i.e. they can not be expressed  
as formal series expansions with leading term being the identity.

Similar observations can be made for the Baxterized solutions coming from the Lyubashenko $r$-matrix.
Let $\tilde R(\lambda) = \lambda I_{V\otimes V} + {\cal P}$ be the Yangian $R$-matrix where ${\cal P} = \sum_{x, y\in X} 
e_{x,y} \otimes e_{y,x}$ is the permutation operator and  
$e_{x, y}$ are the generators of $\mathfrak{gl}_{\cal N}$ in the ${\cal N}$-dimensional representation.
We focus for the sake of simplicity on the finite $\mathfrak{gl}_{\cal N}$ subalgebra of the Yangian; 
$\mathfrak{gl}_{\cal N}$ is a Hopf algebra with counit $\epsilon(e_{x,y}) =0$, $\epsilon(I) =1$, antipode $s(e_{x,y}) = - e_{x,y}$ 
and co-product  $\Delta(e_{x,y})= e_{x,y} \otimes I + I \otimes e_{x,y}.$ 
Moreover, $ \Delta(e_{x, y}) \tilde R(\lambda) = \tilde R(\lambda)  \Delta(e_{x, y}),$ $\forall x,y \in X$,
recalling that $R(\lambda) = {\cal V}_2 \tilde R_{12}(\lambda) {\cal V}_{1}^{-1} = {\cal V}_1^{-1} \tilde R(\lambda) {\cal V}_{2},$
we conclude that $\Delta_j^{(op)}(e_{x,y}) R(\lambda) =R(\lambda) \Delta_j(e_{x,y}),$ $j \in \{1,\ 2\},$ where the
two types of twisted co-products are defined as \cite{DoiSmo2}: $\Delta_1(e_{x,y}) =  e_{x,y} \otimes  I + I \otimes  e_{\tau(x) ,\tau(y)}$ and 
$\Delta_2(e_{x,y}) = e_{\sigma(x), \sigma(y)} \otimes I +I  \otimes e_{x,y},$
($\Delta_2(e_{x, y}) = {\cal V} \otimes {\cal V} \cdot \Delta_1(e_{x, y}) \cdot  {\cal V}^{-1} \otimes {\cal V}^{-1}$). 

Given the above twisted co-products it follows that  
$(\epsilon \otimes \mbox{id})\Delta_1(e_{x,y}) = e_{\tau(x), \tau(y)},$ 
$\ (\mbox{id} \otimes \epsilon)\Delta_1(e_{x,y}) = e_{x,y}$  and 
$(\epsilon \otimes \mbox{id})\Delta_2(e_{x,y}) =e_{x,y},$ 
$\ (\mbox{id} \otimes \epsilon)\Delta_1(e_{x,y}) = e_{\sigma(x),\sigma(y)}$ 
i.e. the twisted co-products  and co-unit do not satisfy the bialgebra axioms. 
In addition, co-associativity for the deformed co-products is also an issue,
indeed it is easily shown using the explicit expressions for the co-products $\Delta_j$: 
$(\mbox{id} \otimes \Delta_1)\Delta_1(e_{x,y}) = 
\Phi^{-1} \cdot (\Delta_1 \otimes  \mbox{id})\Delta_1(e_{x,y})\cdot \Phi$ and 
$(\mbox{id} \otimes \Delta_2)\Delta_2(e_{x,y}) = \hat \Phi^{-1}\cdot (\Delta_2 \otimes  \mbox{id})\Delta_2(e_{x,y})\cdot \hat \Phi$ 
where we define $\Phi = I \otimes I \otimes {\cal V}$ and $\hat \Phi = {\cal V}\otimes I \otimes I$, 
i.e. the co-associativity is now reduced to an almost co-associativity.
Although after applying  the simple twists the underlying algebra is still $\mathfrak{gl}_{\cal N}$ \cite{DoiSmo2}, 
strict  co-associativity as well as the axioms  of the bialgebra involving the co-multiplication and co-unit are not satisfied anymore.

We have briefly demonstrated the intriguing problem of characterizing 
the quantum group associated to set-theoretic solutions as a bialgebra, 
using the simple example of Lyubashenko’s solution. The problem for the general set-theoretic 
solutions of the Yang-Baxter equation
 is studied in \cite{DoikouVlar}.
The general case, also in relation to the Yangian,  as well as the notion of the
quantum double will be addressed  in detail elsewhere.

\subsection*{Acknowledgments}

\noindent  I am grateful to A. Smoktunowicz for reading the manuscript and for useful discussions and comments.
Support from the EPSRC research grants EP/R009465/1  and EP/V008129/1 is also acknowledged.

\end{document}